\documentclass[prl,superscriptaddress,showpacs,twocolumn,longbibliography]{revtex4-1}

\usepackage{braket}
\usepackage{color}
\usepackage[colorlinks=true,linkcolor=blue,urlcolor=blue,citecolor=blue,pdfusetitle]{hyperref}
\usepackage{hyperref}
\usepackage[usenames,dvipsnames]{xcolor}
\usepackage{amsmath,amsthm,amssymb}
\usepackage{graphicx}
\newtheorem{theorem}{Theorem}
\newtheorem{lemma}{Lemma}

\def\(({\left(}
\def\)){\right)}
\def\[[{\left[}
\def\]]{\right]}

\newcommand{\be}{\begin{equation}}
\newcommand{\ee}{\end{equation}}
\newcommand{\ben}{\begin{eqnarray}}
\newcommand{\een}{\end{eqnarray}}
\newcommand{\beq}{\begin{equation}}
\newcommand{\eeq}{\end{equation}}

\begin{document}

\title{Exactness of mean-field equations for open Dicke models with an application to pattern retrieval dynamics}

\author{Federico Carollo}
\affiliation{Institut f\"ur Theoretische Physik, Universit\"at T\"ubingen, Auf der Morgenstelle 14, 72076 T\"ubingen, Germany}
\author{Igor Lesanovsky}
\affiliation{Institut f\"ur Theoretische Physik, Universit\"at T\"ubingen, Auf der Morgenstelle 14, 72076 T\"ubingen, Germany}
\affiliation{School of Physics and Astronomy and \\Centre for the Mathematics and Theoretical Physics of Quantum Non-Equilibrium Systems, University of Nottingham, Nottingham, NG7 2RD, UK}

\date{\today}

\begin{abstract}
Open quantum Dicke models are paradigmatic systems for the investigation of light-matter interaction in out-of-equilibrium quantum settings. Albeit being structurally simple, these models can show intriguing physics. However, obtaining exact results on their dynamical behavior is challenging, since it requires the solution of a many-body quantum system, with several interacting continuous and discrete degrees of freedom. Here, we make a step forward in this direction by proving the validity of the mean-field semi-classical equations for open multimode Dicke models, which, to the best of our knowledge, so far has not been rigorously established. We exploit this result to show that open quantum multimode Dicke models can behave as associative memories, displaying a nonequilibrium phase transition towards a pattern-recognition phase. 
\end{abstract}

\maketitle 

Since its inception \cite{PhysRev.93.99}, the Dicke model has become a paradigm for the study of light-matter interaction and its  equilibrium as well as isolated-system dynamical properties have been widely investigated both theoretically and experimentally  \cite{PhysRevA.7.831,PhysRevA.8.1440,CARMICHAEL197347,PhysRevA.8.2517,PhysRevA.9.418,davies1973,PhysRevA.75.013804,Zhiqiang:17,PhysRevLett.89.253003,PhysRevLett.91.203001,PhysRevLett.104.130401,Baumann:2010aa,PhysRevLett.107.140402}. Nowadays, the interest is in understanding how the presence of an environment, leading to dissipative effects, impacts on the behavior of Dicke models. In this out-of-equilibrium setting, much less is known. Several arguments indicate the persistence of the Dicke superradiant phase transition \cite{doi:10.1002/qute.201800043,10.1371/journal.pone.0235197,PhysRevLett.125.093604,bezvershenko2020}, and this hypothesis is further supported by numerical \cite{PhysRevLett.118.123602} and experimental \cite{Klinder3290} evidence. 

Particularly intriguing is the possibility that these nonequilibrium spin-boson systems can feature dynamics akin to associative memories \cite{john1982neural,Fuchs:1988aa}, i.e.~they can display pattern-recognition behavior \cite{PhysRevLett.107.277201,doi:10.1080/14786435.2011.637980,PhysRevLett.114.143601,PhysRevA.95.032310,Rotondo_2018,PhysRevLett.125.070604}, and implementations of this physics are being explored in realistic experimental setups \cite{marsh2020enhancing}. Couplings between spins and bosons encode different patterns which, in the simplest case, are strings of $\pm1$, see Fig.~\ref{Fig1}(a). The overlap $\xi_\mu$ of the spin configuration with pattern $\mu$, which plays the role of an order parameter, is defined by means of a generalized {\it magnetization} [c.f.~Fig.~\ref{Fig1}(a)]. Assuming the initial configuration to be close to one pattern, two different regimes may emerge. In the first, the state converges --due to dissipation-- to a stationary one where all information about the initial time is lost. As sketched in Fig.~\ref{Fig1}(b), this coincides with a regime where the overlaps $\xi_\mu$ are all zero. In the other, instead, it converges to a stationary state displaying a finite overlap with the initially stored pattern. In this case, the system ``recognizes" the initial condition as a pattern and stores this information in its nonequilibrium steady state. In Dicke models, the observed stationary regime is expected to depend on the spin-boson coupling strength, see Fig.~\ref{Fig1}(b). 

\begin{figure}[t]
\centering
\includegraphics[scale=0.63]{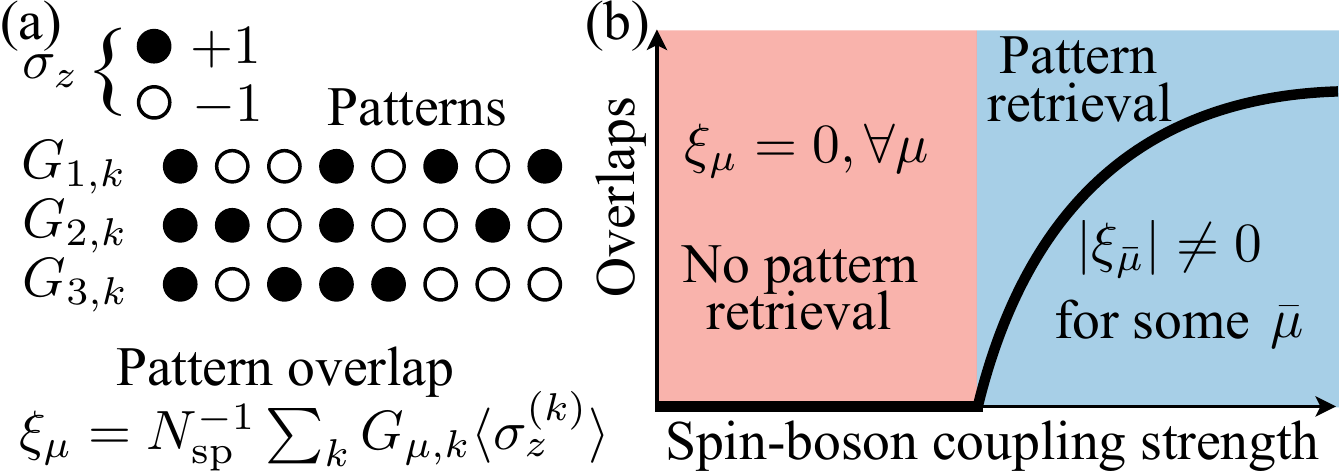}
\caption{{\bf Pattern recognition in Dicke models.} a) Patterns --strings of $\pm 1$-- are encoded in the couplings $G_{\mu,k}$ between $N_{\rm sp}$ spins and $M$ bosonic modes. Each pattern is associated with a mode. The overlap of the quantum state with the patterns is defined as a generalized {\it magnetization} aligned with the coefficients $G_{\mu,k}$. b) As a function of the spin-boson coupling strength, the quantum system passes from a disordered phase, in which it cannot store any pattern, to an ``ordered" one, in which it can recognize and protect a pattern.  }
\label{Fig1}
\end{figure} 

Understanding whether this pattern-recognition behavior corresponds to a genuine nonequilibrium phase requires the study of quantum systems with large number of bosons and spins. Simulations in fully quantum regimes beyond perturbative approaches \cite{PhysRevResearch.2.013198,marsh2020enhancing,PhysRevLett.125.070604} are thus infeasible. Analytically, one may study these sytems relying on so-called mean-field equations, obtained by assuming that expectation values of products of operators factorize \cite{PhysRevLett.118.123602,doi:10.1002/qute.201800043,PhysRevResearch.2.033131}. However, a proof of the validity of this assumption in nonequilibrium open Dicke models is still missing, and a widespread belief is that a ``full quantum treatment" may lead to different results.

In this paper, we provide a proof of the exactness of the mean-field assumption for open multimode Dicke models. This result is relevant as it solves an open question on the validity of the semi-classical treatment for these systems. Further, it allows us to establish the existence of a nonequilibrium pattern-recognition phase transition in Dicke models.  Our proof --which takes inspiration from Ref.~\cite{Pickl}-- is of broad applicability: it can be adapted to account for the presence of individual spin dissipative processes \cite{PhysRevLett.118.123602}, to account for time-dependent coefficients in the generator \cite{Niedenzu_2018,PhysRevLett.124.170602,carollo2020nonequilibrium}, or even to other models with all-to-all couplings  \cite{Bagarello:1992aa,Benatti_2018,PhysRevLett.121.035301,Norcia259,PhysRevLett.123.260401,PhysRevB.101.214302}, also with multi-body interactions \cite{wang2020dissipative,Grimsmo_2013,Garbe:2020aa}, \\

{\it \bf Open multimode Dicke models.---} Our Dicke model consists of an ensemble of $N_{\rm sp}$ spins coupled to $M$ different bosonic modes, described by annihilation and creation operators $a_\mu,a_\mu^\dagger$ obeying canonical commutation relations \cite{petz1990invitation}. Spins are two-level systems, with excited state $\ket{\bullet}$ and ground state $\ket{\circ}$. Transitions between  states in the $k$-th spin are implemented by the Pauli operator $\sigma_{ x}^{(k)}$, where $\sigma_{ x}\ket{\bullet/\circ}=\ket{\circ/\bullet}$. The operator $\sigma_{ z}^{(k)}$, with $\sigma_{ z}\ket{\bullet}=\ket{\bullet}$ and $\sigma_{ z}\ket{\circ}=-\ket{\circ}$, indicates the presence of an excitation. We also define $\sigma_{ y}^{(k)}=-i\sigma_{ z}^{(k)}\sigma_{ x}^{(k)}$.

The (Markovian) nonequilibrium dynamics of the spin-boson model is implemented by the Lindblad generator $\dot{X}=\mathcal{L}[X]$ \cite{lindblad1976,Gorini1976,breuer02a}, providing the time-evolution of a generic operator $X$. Defining $n_\mu=a^\dagger_\mu a_\mu$, we consider
\begin{equation} 
\mathcal{L}[X]:=i[H,X]+\sum_{\mu=1}^M \kappa_\mu\left(a^\dagger_\mu  X a_\mu -\frac{1}{2}\left\{n_\mu, X\right\}\right)\, ;
\label{MT-Lind}
\end{equation}
the second term appearing on the right-hand side describes boson losses, at rate $\kappa_\mu$ for the different modes, while $H$ is the system Hamiltonian. This operator consists of a free contribution for both spins and bosons
$$
H_{\rm F}=\Omega \sum_{k=1}^{N_{\rm sp}}\sigma_{ x}^{(k)}+\sum_{\mu=1}^M\Omega_\mu \, n_\mu \, , 
$$
and of an interaction term  
\begin{equation}
H_{\rm int}=\frac{g_0}{\sqrt{N_{\rm sp}}}\sum_{\mu=1}^M \sum_{k=1}^{N_{\rm sp}} G_{\mu,k}\left(a_\mu +a_\mu^\dagger \right)\sigma_{ z}^{(k)}\, .
\label{orig-int}
\end{equation}
The coefficients $G_{\mu,k}$ specify the spin-boson interaction. We consider these to be independent identically distributed random variables assuming the values $+1$ or $-1$ with equal probability, as sketched in Fig.~\eqref{Fig1}(a). The scaling $1/\sqrt{N_{\rm sp}}$ --typical for these models-- is important to establish a well-defined thermodynamic limit \cite{doi:10.1002/qute.201800043}  (see also \cite{doi:10.1098/rspa.2017.0856} for an application to open systems). For each $\mu$, the string $G_{\mu,k}$ forms a pattern which is encoded in the Hamiltonian. A key result of this paper consists in showing that the system can recognize and protect an initially stored pattern, for strong enough spin-boson coupling $|g_0|$, see Fig.~\ref{Fig1}. 

\begin{figure}[t]
\centering
\includegraphics[scale=0.54]{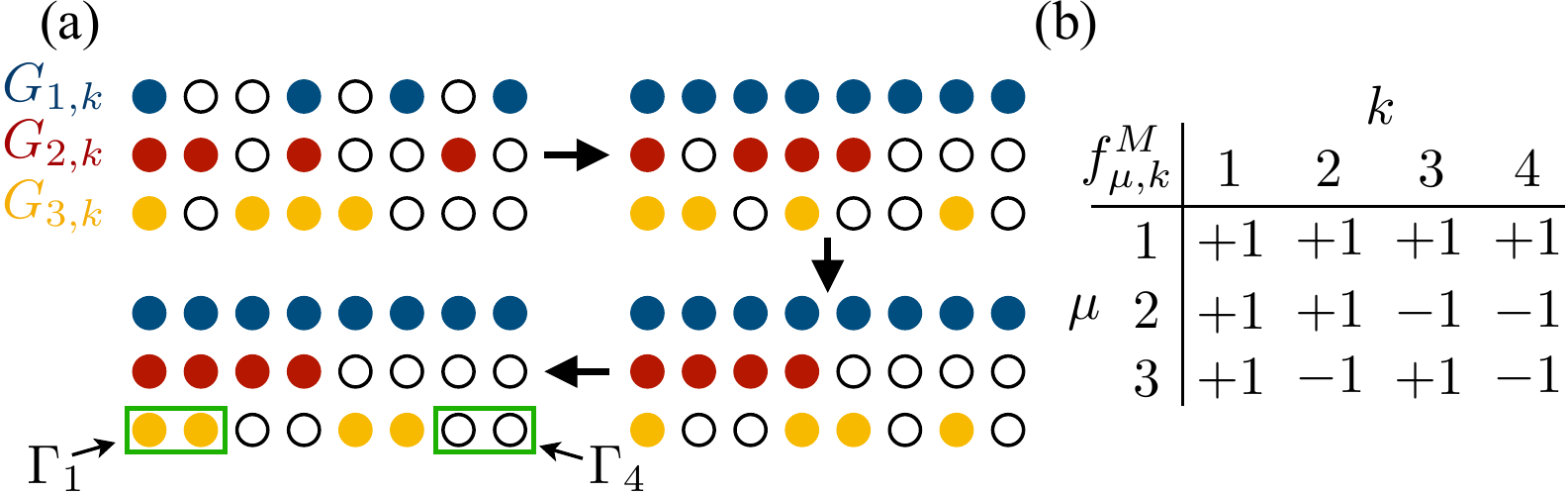}
\caption{{\bf Mapping to large spins.} a) Example of the mapping for $M=3$ patterns and $N_{\rm sp}=8$ spins. The original coupling between the $\mu$-th mode and the $k$-th spin is encoded in $G_{\mu,k}$. To perform the mapping we first apply a gauge transformation making $G_{1,k}=1, \, \forall k$. Then, we reorder $G_{2,k}$ to put all $+1$ first. Finally, also the last pattern is reordered by moving the $+1$ towards the right and the $-1$ towards the left in each sub-block identified by the new $G_{2,k}$. In this way, $2^{M-1}$ subsets of spins $\Gamma_{ k}$, equally coupled with each mode, are identified ($\Gamma_1$ and $\Gamma_4$ are highlighted in the figure for clarity). b) These subsets of spins are described by ``large-spin" operators and couple to bosons as specified by the matrix $f_{\mu, k}^M$. }
\label{Fig2}
\end{figure} 

Before showing this, we make some considerations which bring the model into a convenient form, see Fig.~\ref{Fig2}. First, without loss of generality, the first pattern, $G_{1,k}$, which is made of $\pm 1$, can be brought into a pattern with all $+1$, by means of the gauge transformation $\sigma_{ z}\to -\sigma_{ z}$ applied to those spins $h$ for which, originally, $G_{1,h}=-1$. Then, we reorder the remaining $M-1$ rows of $G_{\mu,k}$. We look at $G_{2,k}$: this has $\pm 1$ at random positions. We now relabel the spins. We take those with $G_{2,h}=+1$ to the left and those with $G_{2,h}=-1$ to the right. This reshaping is not affecting the first pattern. In addition, there is a $\tilde{k}$ such that for $k\le \tilde{k}$, $G_{2,k}=1$ while $G_{2,k}=-1$ otherwise. We then move to $G_{3,k}$ and we relabel spins as follows. In the subset of spins for which $G_{2,k}=1$, we have values of $G_{3,k}$ which can be both positive and negative. We thus reorder this subsequence in such a way that all $+1$ are moved on the left and $-1$ on the right. The same can be done for the subset of the sequence $G_{3,k}$ corresponding to values $G_{2,k}=-1$. This procedure, sketched in Fig.~\ref{Fig2}, is then iterated up to the last pattern. 

This mapping generates $2^{M-1}$ subsets of spins, described by ``large-spin" operators and interacting with the bosonic modes. For $N_{\rm sp}\gg1$, these subsets are expected to have the same number of spins. This is due to the fact that, given the statistical properties of the $G_{\mu,{ k}}$, in a large enough set of randomly chosen spins there is, at leading order in extensivity of the set, an equal number of $+1$ and of $-1$, in their $G_{\mu,k}$. We can thus consider subsets to contain $N=N_{\rm sp}/2^{M-1}$ spins. In this representation, the interaction Hamiltonian reads 
\begin{equation}
H_{\rm int}^N=\frac{g}{\sqrt{N}}\sum_{\mu=1}^M \sum_{{ k}=1}^{2^{M-1}} f_{\mu,{ k}}^M\left(a_\mu +a_\mu^\dagger \right)S_{ z,k}\, ,
\label{new-int}
\end{equation}
where $S_{ a,k}=\sum_{h\in \Gamma_{ k}}\sigma_{ a}^{(h)}$ is the sum of the $\sigma_{ a}$-spin operators which belongs to the ${ k}$-th subset, denoted as $\Gamma_{ k}$ [see Fig.~\ref{Fig2}(a)]. In addition, we have defined $g=g_0/\sqrt{2^{M-1}}$. The coefficients $f_{\mu,{ k}}^M=\pm1$ specify the interaction between spins in $\Gamma_{ k}$ and the $\mu$-th boson. This representation provides a more compact formulation of the model.  This mapping can be extended to consider models whose spin-only part of the dynamical generator is not invariant under the gauge transformation or also to consider generic distributions for  $G_{\mu, k}$ \cite{SM}. \\

{\it \bf Mean-field dynamics.---} As a consequence of the previous mapping, it is sufficient for understanding the behaviour of our nonequilibrium Dicke model to focus on the dynamics of the ``large-spin" operators. In this representation, the generator is $\mathcal{L}_N$, the same as the one in Eq.~\eqref{MT-Lind} with Hamiltonian rewritten as $H_N=H_{\rm F}+H_{\rm int}^N$. The expectation of time-evolved operators, $X_t=e^{t\, \mathcal{L}_N}[X]$, is given by 
$
\langle X\rangle_t=\omega_t\left(X\right):=\omega\left(e^{t\, \mathcal{L}_N}[X]\right)
$,
where the functional $\omega$ represents the initial state, while $\omega_t$ the time-evolved one. As a consequence, we have
\begin{equation}
\dot{\omega}_t\left(X\right)=\omega_t\left(\mathcal{L}_N\left[X\right]\right)\, .
\label{t-der-state}
\end{equation}
We are interested in the ``macroscopic" operators \cite{Lanford:1969aa,strocchi2005symmetry,verbeure2010many,bratteli2012operator}
\begin{equation}
m_{ a,k}^N:=\frac{1}{N}S_{ a,k}\, , \,\, \mbox{for } { a}={ x}, { y}, { z}\, , \quad \alpha_{\mu,N}:=\frac{a_\mu}{\sqrt{N}}\, ;
\label{seq-spin}
\end{equation}
the first ones are the usual average ``magnetization" operators of the spin ensembles, while the rescaled bosonic operators appear typically in superradiant transitions. Indeed, a non-vanishing expectation of these operators implies a macroscopic ($\propto N$) bosonic occupation. 

We want to derive the dynamics of these quantum operators in the thermodynamic limit $N,N_{\rm sp}\to\infty$. We thus compute the action of the generator $\mathcal{L}_N$ on the operators in Eq.~\eqref{seq-spin} and get \cite{SM}
\begin{equation}
\begin{split}
\mathcal{L}_N\!\!\left[m_{ a,k}^N\right]&\!\!=\!\!\!\sum_{ b}\!\!\left(\!\!-2\Omega\epsilon_{ xab}
\!\!-\!2g\!\!\sum_{\mu}\!\epsilon_{ zab}f_{\mu,{ k}}^M\!\left(\!\alpha_{\mu,N}^\dagger \!+\!\alpha_{\mu,N}\!\right) \!\!\right)\!\!m_{ b,k}^N\\
\mathcal{L}_N\!\!\left[\alpha_{\mu,N}\right]&\!\!=\!\!-\left(i\Omega_\mu+\frac{\kappa_\mu}{2}\right)\alpha_{\mu,N}-ig\sum_{{ k}=1}^{2^{M-1}}f_{\mu,{ k}}^M m_{ z,k}^N\, ,
\end{split}
\label{heis-eq}
\end{equation}
where $\epsilon_{ abc}$ is the fully anti-symmetric tensor. 
To make progress, one typically assumes that the dynamics does not generate correlations among the different constituents in the thermodynamic limit, so that expectation values factorize. This leads to the mean-field equations
\begin{equation}
\begin{split}
\dot{m}_{ a,k}&=-2\Omega\!\sum_{b}\!\epsilon_{ xab}m_{ b,k}\!-2 g\!\sum_{ b,\mu}\!\epsilon_{ zab}f_{\mu,{ k}}^M\left(\alpha_{\mu}^\dagger +\alpha_{\mu}\right) m_{ b,k}\, ,\\
\dot{\alpha}_\mu&=-\left(i\Omega_\mu+\frac{\kappa_\mu}{2}\right)\alpha_\mu-ig\sum_{{ k}=1}^{2^{M-1}} f_{\mu,{ k}}^M m_{ z,k}\, .
\end{split}
\label{lim-eqs}
\end{equation}
In order to show that they are exact in the thermodynamic limit, we need to prove that 
\begin{equation}
\lim_{N\to\infty}\omega_t\left(m_{ a,k}^N\right)-m_{ a,k}(t)=0=\lim_{N\to\infty}\omega_t\left(\alpha_{\mu,N}\right)-\alpha_\mu(t)\, ,
\label{limits}
\end{equation}
meaning that the expectation of the operators of Eqs.~\eqref{seq-spin} behaves, for large $N$, as the time-dependent scalar functions $m_{ a,k}(t),\alpha_\mu(t)$ obeying Eqs.~\eqref{lim-eqs}. To obtain this result, a proper strategy must be identified. In particular, an appropriate ``cost function" controlling the above limits is needed. Defining $E_{ a,k}=m_{ a,k}^N-m_{ a,k}(t)$ and $A_{\mu}=\alpha_{\mu,N}-\alpha_\mu(t)$, we consider
\begin{equation}
\begin{split}
\mathcal{E}_N(t):=\!\!\!\!\!\!\!\!\sum_{ k=1,a=x,y,z}^{2^{M-1}}\!\!\!\!\!\!\!\omega_t \left(E_{ a,k}^2\right)+\sum_{\mu=1}^M\omega_t \left(A_\mu^\dagger A_\mu+A_\mu A_\mu^\dagger \right)\, .
\end{split}
\label{error}
\end{equation}
This quantity is a sum of positive contributions consisting of the expectation of the square of the distance of the operators from their mean-field counterpart. Namely, $\mathcal{E}_N(t)$ measures the fraction of spins or bosons not behaving as dictated by Eqs.~\eqref{lim-eqs}. In addition, via Cauchy-Schwarz inequality, one can show that
\begin{equation}
\left|\omega_t \left(E_{ a,k}\right)\right|\le \sqrt{\omega_t \left(E_{ a,k}^2\right)}\le \sqrt{\mathcal{E}_N(t)}\, ,
\label{err-bound}
\end{equation}
and thus $\lim_{N\to\infty}\mathcal{E}_N(t)$ controls the limits in Eq.~\eqref{limits}, as desired. For physical initial states \cite{strocchi2005symmetry,verbeure2010many,bratteli2012operator}, with short-range correlations,  one has $\lim_{N\to\infty}\mathcal{E}_N(0)=0$. As we now show, for these states, $\mathcal{E}_N(t)$ vanishes for large $N$, implying the exactness of the mean-field assumption for these nonequilibrium multimode Dicke models. \\

\noindent {\bf Theorem.} With the above definitions, if the initial state of the system is such that $\lim_{N\to\infty}\mathcal{E}_N(0)=0$ then, for all finite $t$, we have that $\lim_{N\to\infty}\mathcal{E}_N(t)=0$.

{\it Proof:} The full proof is reported in Ref.~\cite{SM}. Here we provide the main steps. The idea is to use Gronwall's Lemma \cite{zbMATH02606740,bellman1943}, which states that if a positive, bounded, and $N$-independent constant $C$, such that $\dot{\mathcal{E}}_N(t)\le C\,  \mathcal{E}_N(t)$, exists then 
\begin{equation}
\mathcal{E}_N(t)\le e^{C\, t}\mathcal{E}_N(0)\, .
\label{G-lemma}
\end{equation}
With the assumption $\lim_{N\to\infty}\mathcal{E}_N(0)=0$, letting $N\to\infty$ in the above relation would prove the theorem. What is missing is to show that such constant $C$ indeed exists. This can be achieved by directly inspecting the time derivative of all terms forming $\mathcal{E}_N(t)$. They are given by sums of contributions having, for instance, the form $\omega_t\left(E_{ b,k}B A_\mu\right)$, where $B$ can either be an operator or a scalar from Eqs.~\eqref{lim-eqs}. In addition, it can be shown that 
$$
\left|\omega_t\left(E_{ b,k}B A_\mu\right)\right|\le \left\|B\right\|\mathcal{E}_N(t)\, ,
$$
and this gives a way to estimate a suitable constant $C$. We thus obtain 
$$
\frac{d}{dt}\mathcal{E}_N(t)\le \left|\frac{d}{dt}\mathcal{E}_N(t)\right|\le C\mathcal{E}_N(t)\, ,
$$
and we can exploit Gronwall's Lemma to finish the proof of the theorem as already discussed. \qed\\

{\bf Pattern-recognition phase transition.---} With the above result, we establish that the semi-classical mean-field equations \eqref{lim-eqs} correctly capture the behavior of our system, in the thermodynamic limit. As such, we can now use these equations to unveil the presence of a nonequilibrium pattern-recognition phase transition.

In the original formulation of the problem, see Eq.~\eqref{orig-int} and Fig.~\ref{Fig1}(a), we can define the overlap of the quantum state of the spins with the pattern $\mu$ as
$$
\xi_\mu:=\lim_{N_{\rm sp}\to\infty}\frac{1}{N_{\rm sp}}\sum_{k=1}^{N_{\rm sp}}G_{\mu,k}\braket{\sigma_{ z}^{(k)}}_t\, .
$$
This equation shows that, if the expectaction value of the operator $\sigma_{ z}$ is, for each spin, aligned with the corresponding value of $G_{\mu,k}$, then the overlap $|\xi_\mu|$ is different from zero (pattern retrieval). Otherwise, $\xi_\mu$ tends to vanish for $N_{\rm sp}\to\infty$ (pattern not retrieved). In the large-spin representation, the overlaps can be expressed in terms of the coefficients $f_{\mu,{ k}}^M$ and of the macroscopic operators $m_{ z,k}^N$, [c.f. Eq~\eqref{new-int} and Fig.~\ref{Fig2}]. In particular,
$$
\xi_\mu=\frac{1}{2^{M-1}}\sum_{ k=1}^{2^{M-1}}f_{\mu,{ k}}^M\lim_{N\to\infty}\omega_t\left(m_{ z,k}^N\right)\, .
$$
Invoking our theorem, we can thus study the dynamics and the stationary properties of these overlaps through the scalars $m_{ z,k}$, obeying the mean-field equations \eqref{lim-eqs}.

To prove the existence of the phase transition, we first show the presence of different stationary solutions to Eq.~\eqref{lim-eqs}, featuring a finite overlap with one of the patterns. Without loss of generality, we consider all rates of the dynamical generator to be positive and, further, that the constant of motion $m_{ T,k}^2=\sum_{ a}m_{ a,k}^2=1$, $\forall { k}$. Then, we take the ansatz solution $m_{ z,k}=f_{\nu,{ k}}^M|z|$, aligned with pattern $\nu$, and look for conditions ensuring its existence as a stationary solution for Eqs.~\eqref{lim-eqs}. Note that such ansatz has indeed a finite overlap with pattern $\nu$, since $\xi_\nu=|z|$ while $\xi_\mu=0$  $\forall \mu\neq\nu$, and that 
also $m_{ z,k}=-f_{\nu,{ k}}|z|$ would be valid, with $\xi_\nu=-|z|$.

\begin{figure}[t]
\centering
\includegraphics[scale=0.53]{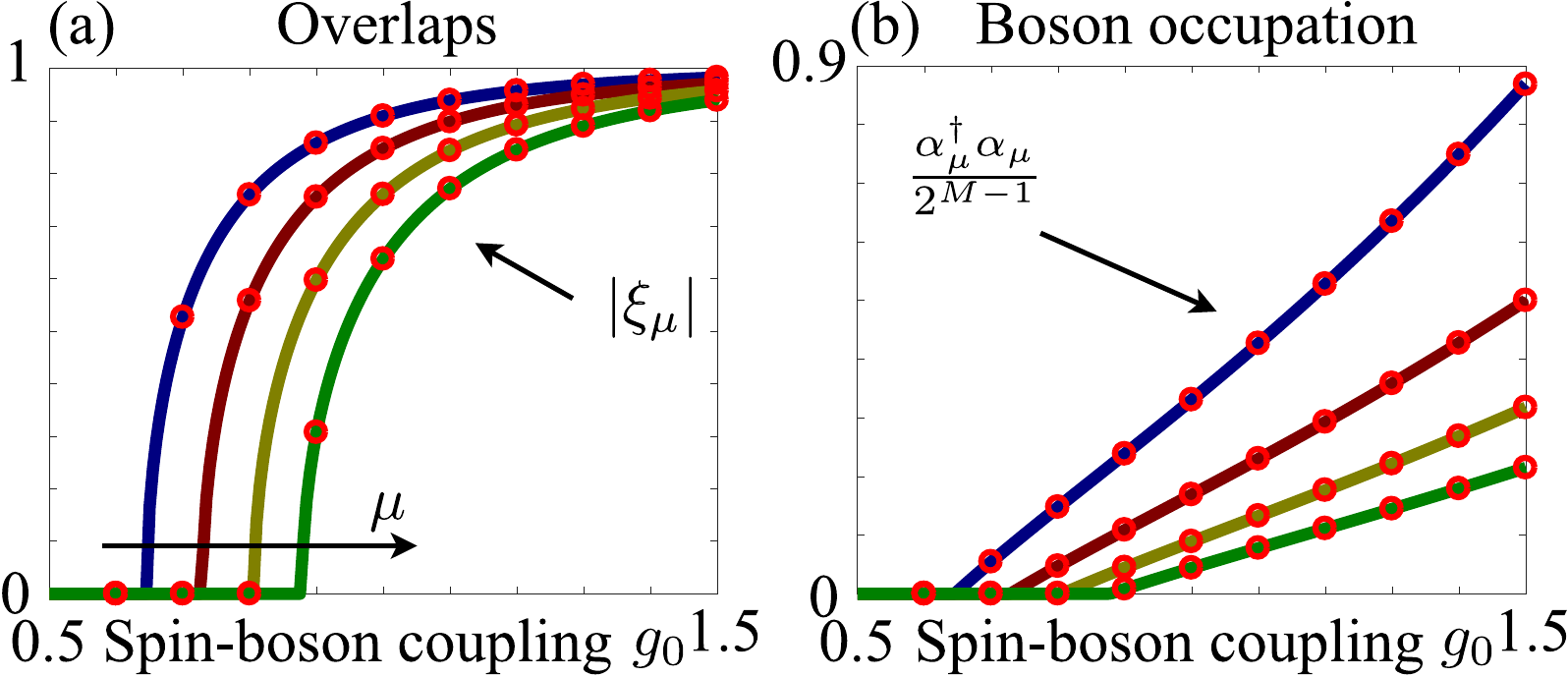}
\caption{{\bf Pattern-recognition phase transition.} Comparison between theoretical prediction (solid lines) and numerical simulations of the mean-field equations (circles). We consider $M=4$. a) Each curve corresponds to the stationary overlap $\left|\xi_\mu\right|$ computed from the initial condition $\xi_\mu=1$ as a function of $g_0$, for $\Omega=0.5$, $\Omega_\mu=\Omega(\mu+2)$. Rates are in units of $\kappa$. Different colors correspond to values of $\mu$ growing as indicated by the arrow. Both theoretical and numerical results display a nonequilibrium phase transition, as shown by the behavior of the overlap. b) Same parameters and same order for the curves as in a). The occupation of the $\mu$-th bosonic mode becomes macroscopically occupied when the corresponding pattern is stored in the stationary state.}
\label{Fig3}
\end{figure} 

By substituting the ansatz for $m_{ z,k}$ in Eqs.~\eqref{lim-eqs}, taking $m_{ y,k}=0$ and appropriately fixing the values of $m_{ x,k}$ (see Ref.~\cite{SM} for details) we find that the relation
\begin{equation}
|z|=\sqrt{1-\frac{1}{4g_0^4}\left(\frac{\Omega}{\Omega_\nu}\right)^2\left[\Omega_\nu^2+\left(\frac{\kappa_\nu}{2}\right)^2\right]^2}\, 
\label{z-stat}
\end{equation}
must be satisfied, in order for the assumed stationary solution to exist. This is not always the case; indeed, $|z|$ must be a positive real number, $|z|\in[0;1]$, and this only happens if the argument of the square root is positive. This observation yields a critical value,
$$
g_{\rm crit}=\sqrt{\frac{1}{2}\left(\frac{\Omega}{\Omega_\nu}\right)\left[\Omega_\nu^2+\left(\frac{\kappa_\nu}{2}\right)^2\right]}\, ,
$$
such that for $g_0\ge g_{\rm crit}$ the ansatz solution exists, with $|z|$ given by Eq.~\eqref{z-stat}. On the other hand, if $g_0<g_{\rm crit}$, we can only have $|z|=0$, and we are outside the pattern-recognition phase. The critical $g$ depends on the pattern through the parameters $\Omega_\nu,\kappa_\nu$, see also Fig.~\ref{Fig3}(a-b). Further, note that a finite stationary overlap corresponds to a macroscopic occupation of the associated bosonic mode. Our theorem indeed implies  $N^{-1}\langle a^\dagger_\mu a_\mu\rangle\to \left|\alpha_\mu\right|^2$, for  $N\to\infty$, and we have $\left|\alpha_\mu\right|\propto \left|\xi_\mu\right|$ \cite{SM}. This feature, shown in Fig.~\ref{Fig3}(b), establishes a connection between pattern-recognition and the superradiant phase transitions in open multimode Dicke models. \\

{\bf Discussion.---} We have derived two key results for multimode Dicke models. First, we have shown that the mean-field assumption, typically exploited to consider the large-scale behavior of these systems, actually provides an exact description in the thermodynamic limit. Second, we have used this new insight to reveal the presence of a nonequilibrium phase transition from a disordered phase to a pattern-recognition phase in open multimode Dicke models. The stability of stationary solutions, such as the one of Eq.~\eqref{z-stat}, for open Dicke models has been shown, for instance, in Refs.~\cite{doi:10.1002/qute.201800043,PhysRevLett.118.123602}. For the multimode settings investigated here, the agreement of our numerical results with analytical ones [c.f.~Fig.~\ref{Fig3}] suggests that the proposed stationary states, having finite overlap with the patterns, possess stable basins of attraction in the pattern-recognition phase. Interestingly, the critical spin-boson coupling strength depends on the specific pattern through the corresponding bosonic mode parameters. This may allow for intermediate regimes of pattern recognition, where only certain patterns can be stored and retrieved. 

\noindent Following Ref.~\cite{Benatti_2018}, we remark that the validity  of the semi-classical Eqs.~\eqref{seq-spin} provides a necessary ingredient to obtain mathematically rigorous results on quantum fluctuations. It would be interesting to exploit it to re-obtain bosonic descriptions \cite{PhysRevLett.90.044101,PhysRevE.67.066203} employed for the investigation of quantum fluctuations in closed Dicke models and to extend these to open systems, via quantum central limit theorems \cite{Goderis:1989aa,verbeure2010many,Benatti_2018}. Contrary to Holstein-Primakoff approximations, these procedures do not assume a conserved total spin operator and are thus more general \cite{doi:10.1002/qute.201800043}. \\

\begin{acknowledgments}
\noindent We acknowledge support from the ``Wissenschaftler-R\"uckkehrprogramm GSO/CZS" of the 
Carl-Zeiss-Stiftung and the German Scholars Organization e.V., as well as through the
Deutsche Forschungsgemeinsschaft (DFG, German Research Foundation) under Project No. 
435696605, and under  Germany's  Excellence  Strategy - EXC No. 2064/1 - Project No. 390727645. 
FC acknowledges support through a Teach@T\"ubingen Fellowship.
\end{acknowledgments}

\bibliography{Notes_BIBLIO}

\newpage

\renewcommand\thesection{S\arabic{section}}
\renewcommand\theequation{S\arabic{equation}}
\renewcommand\thefigure{S\arabic{figure}}
\setcounter{equation}{0}
\setcounter{figure}{0}

\onecolumngrid

\begin{center}
{\Large SUPPLEMENTAL MATERIAL}
\end{center}
\begin{center}
\vspace{0.8cm}
{\Large Exactness of Mean-Field Equations for Open Dicke Models with an Application to Pattern Retrieval Dynamics}
\end{center}
\begin{center}
Federico Carollo,$^{1}$ and Igor Lesanovsky$^{1,2}$
\end{center}
\begin{center}
$^1${\em Institut f\"ur Theoretische Physik, Universit\"at T\"ubingen,}\\
{\em Auf der Morgenstelle 14, 72076 T\"ubingen, Germany}\\
$^2${\em School of Physics and Astronomy and}\\
{\em Centre for the Mathematics and Theoretical Physics of Quantum Non-Equilibrium Systems,}\\
{\em  University of Nottingham, Nottingham, NG7 2RD, UK}\\

\end{center}

\section{Possible extensions of the mapping}
As briefly mentioned in the main text, the mapping that we have introduced is applicable to more general models than the one we have focussed on in this work. We discuss here two possible extensions. \\

\noindent The first step of the mapping, as presented in the main text, is the gauge transformation $\sigma_{ z}\to -\sigma_{ z}$. This transformation may also, in general, modify the spin Hamiltonian (or a possible dissipative contributions on the spins) or may lead to complications.  However, such a step is not necessary. Indeed, instead of applying the gauge transformation to the spins in order to make the first pattern uniform ($G_{1,k}=1$, $\forall k$), one can directly start acting on the first pattern reordering it by moving all spins with $G_{1,k}=+1$ to the left and all spins with $G_{1,k}=-1$ to the right. After this is done, one can then proceed to reorder analogously the other patterns. In this way, instead of $2^{M-1}$ subsets of spins, the mapping generates $2^M$ of them. An illustration of this version of the mapping is given in Fig.~\ref{Fig1SM}.

\begin{figure}[h]
\centering
\includegraphics[scale=0.74]{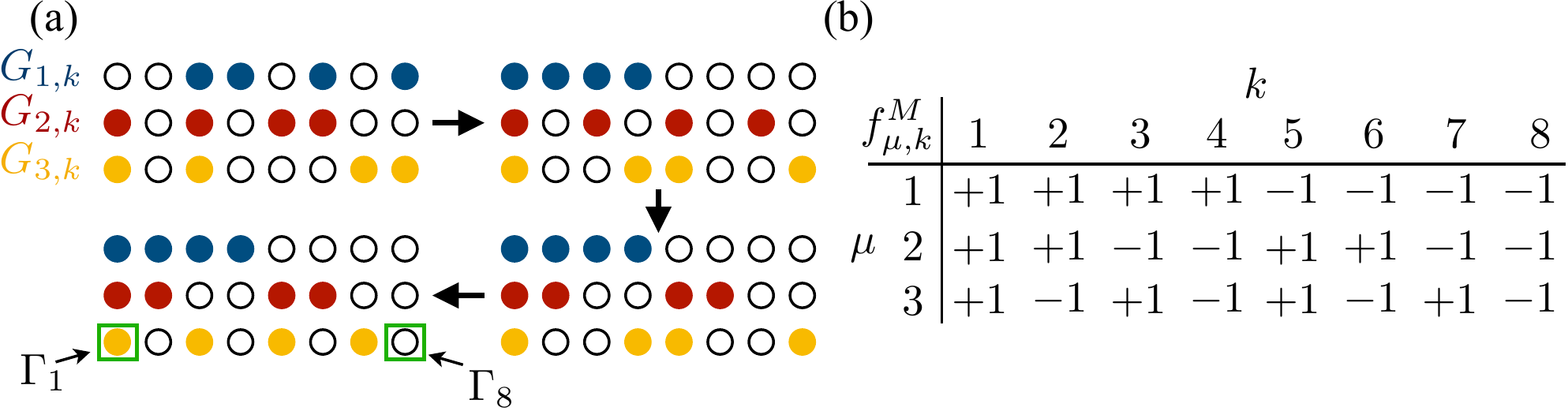}
\caption{{\bf Mapping to large spins without gauge transformation.} a) Example of the mapping without initial gauge transformation for $M=3$ patterns and $N_{\rm sp}=8$ spins. The coupling between the $\mu$-th mode and the $k$-th spin is encoded in the coefficients $G_{\mu,k}$. In this version of the mapping, we do not make the initial gauge transformation but rather start reordering the first pattern $G_{1,k}$. We relabel spins in such a way that all $G_{1,k}=1$ appear first, and are followed by all the $G_{1,k}=-1$. This merely amounts to a relabeling of the spins. We then proceed to reorder also the other patterns in an analogous way to what done in the main text. We notice that, in this way,  the mapping produces $2^{M}$ subsets of spins $\Gamma_{ k}$, which are equally coupled to each mode, ($\Gamma_1$ and $\Gamma_8$ are highlighted in the figure for clarity). b) These subsets of spins couple to bosons as specified by the matrix $f_{\mu, k}^M$. }
\label{Fig1SM}
\end{figure} 

\noindent Another possible extension considers different probability distributions for the coupling coefficients $G_{\mu,k}$. In general, each $G_{\mu,k}$ may assume the value $x_i$, where $i=1,2,\dots, d$, with probabilities $p_i$. In this case, one can reorder the first pattern by relabelling all those spins $k$, associated with a $G_{1,k}=x_1$ to move them to the first positions, followed by those with $G_{1,k}=x_2$, and so on till the spins having $G_{1,k}=x_d$ have been accomodated.  As for the mapping in the main text, this is just a simple relabelling of the spins. Then, moving to the second pattern $G_{2,k}$, we can proceed as follows. We focus, one by one, on the subsets of spins $k$ for which $G_{1,k}=x_i$. Within these subsets we further reorder the spins, relabelling them in such a way that the ones having $G_{2,h}=x_1$ are moved to the first positions in the subset, followed by those with $G_{2,h}=x_2$, till those having $G_{2,h}=x_d$. For the third, as well as for the remaining patterns the spins are analogously reordered. \\

\noindent We notice that in this case, the procedure generates $d^M$ subsets of spins which are, within each subset, all equally coupled to the different bosonic modes. As such, these ensembles can be treated as collective large spins. Interestingly, the extensivity of these subsets is, in general, not uniform and in fact depends on the probabilities $p_i$ of extracting the different couplings. In particular, we have that if the spins in the subset $\Gamma_{k}$ interact with the $\mu$ bosonic mode with the coupling $x_{i_\mu}$, then, for large numbers of total spins, $N_{\rm sp}\gg 1$, the number of spins forming the subset $\Gamma_{ k}$ is given by $N_{ k}\approx p_{i_1}p_{i_2}\dots p_{i_M}N_{\rm sp}$. When all subsets are expected to be different, the proof that we present is still applicable; it is sufficient to redo the same steps using $N_{\rm sp}$ instead of the $N$ introduced in the text, also in the definition of the macroscopic operators \eqref{seq-spin}. The extensivity of the different ensembles will then pose constraints on the modulus of the expectation values of the different macroscopic operators. For instance, in a given ensemble $\Gamma_{k}$ made of $N_{ k}\approx p_{i_1}p_{i_2}\dots p_{i_M}N_{\rm sp}$ spins, expectation values of  macroscopic spin operators cannot be larger than $p_{i_1}p_{i_2}\dots p_{i_M}$, in the thermodynamic limit. 

\begin{figure}[t]
\centering
\includegraphics[scale=0.74]{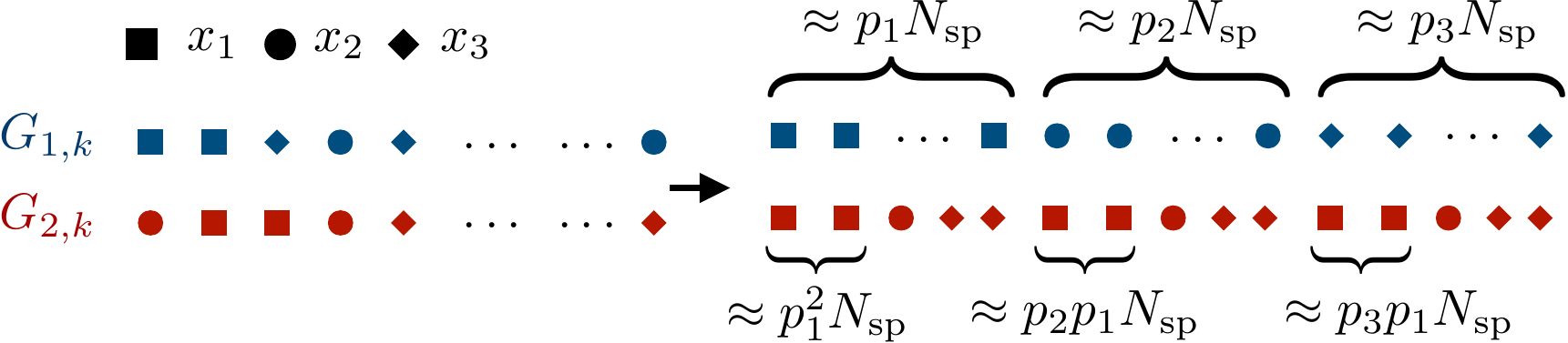}
\caption{{\bf Mapping to large spins for general distributions.} We present here an example of the mapping for couplings $G_{\mu,k}$ which can take three possible values $x_1,x_2,x_3$ and for $M=2$ bosonic modes. The idea is to relabel the spins in such a way that the two patterns $G_{1,k}$ and $G_{2,k}$ are reordered as in the figure. In particular, the first pattern will be reordered in such a way that the coupling $x_1$ is associated with the first subset of spins, $x_2$ with the second and $x_3$ with the third. Each of these subsets can be partitioned in smaller subsets. For instance, spins which couple through $x_1$ with the first mode, can be reordered according to their coupling with the second mode. The reordering of this last pattern reveals $3^2$ ensembles of spins which are all coupled with both bosonic modes in the same way. Since the probabilities of having the different values $x_i$ for the couplings are not uniform, the extensivity of the different ensembles is modulated by the probabilities $p_i$ as explained in the main text. }
\label{Fig2SM}
\end{figure}

\section{Proof of Main Theorem}
\begin{theorem}
Given the sequence of generators $\mathcal{L}_N$ introduced in the main text, the sequences of operators in Eq.~\eqref{seq-spin}, and the scalar time-dependent functions appearing in Eqs.~\eqref{lim-eqs}, the error $\mathcal{E}_N(t)$, defined as 
\begin{equation}
\begin{split}
\mathcal{E}_N(t):&=\sum_{\mu=1}^M\left[\omega_t\left(\left[\alpha_{\mu,N}-\alpha_\mu(t)\right]^\dagger \left[\alpha_{\mu,N}-\alpha_\mu(t)\right]\right)+\omega_t\left(\left[\alpha_{\mu,N}-\alpha_\mu(t)\right] \left[\alpha_{\mu,N}-\alpha_\mu(t)\right]^\dagger\right)\right]\\
&+\sum_{ b=x,y,z}\sum_{{ k}=1}^{2^{M-1}}\omega_t\left(\left[m_{ b,k}^N-m_{ b,k}(t)\right]^2\right)\, ,
\end{split}
\label{def-E}
\end{equation}
is such that
\begin{equation}
\lim_{N\to\infty}\mathcal{E}_N(t)=0\, ,
\label{theorem1-1}
\end{equation}
$\forall t\ge0$ finite, if the initial condition
\begin{equation}
\lim_{N\to\infty}\mathcal{E}_N(0)=0\, ,
\label{theorem1-2}
\end{equation}
is satisfied. 

Eq.~\eqref{theorem1-1} implies that the mean-field equations for the scalar quantities $\alpha_\mu$ and $m_{a,k}$, for all $\mu,{ k}$ and ${a=x,y,z}$, correctly describe the dynamics of the expectation values of the limiting operators of the sequences in Eqs.~\eqref{seq-spin}.
\label{theorem}
\end{theorem}

\begin{proof}
In order to prove the theorem, we start considering the time-derivative of $\mathcal{E}_N(t)$
\begin{equation}
\begin{split}
\frac{d}{dt}\mathcal{E}_N(t)=&\sum_{\mu=1}^M\left[\frac{d}{dt}\omega_t\left(\left[\alpha_{\mu,N}-\alpha_\mu(t)\right]^\dagger\left[\alpha_{\mu,N}-\alpha_\mu(t)\right]\right)+\frac{d}{dt}\omega_t\left(\left[\alpha_{\mu,N}-\alpha_\mu(t)\right]\left[\alpha_{\mu,N}-\alpha_\mu(t)\right]^\dagger\right)\right]+\\
+&\sum_{ b=x,y,z}\sum_{ k=1}^{2^{M-1}}\frac{d}{dt}\omega_t\left(\left[m_{ b,k}^N-m_{b,k}(t)\right]^2\right)\, .
\end{split}
\label{theorem-partial1}
\end{equation}
Using the results in Lemmata \ref{Lemma3} and \ref{Lemma4} to bound the modulus of all terms appearing in the sums, we find
\begin{equation}
\frac{d}{dt}\mathcal{E}_N(t)\le \left(2Md_0+2^{M-1}3c_0\right)\mathcal{E}_N(t)\, ,
\label{theorem-partial2}
\end{equation}
where $c_0,d_0$ are the time-independent and $N$-independent bounded positive quantities defined in Lemmata \ref{Lemma3} and \ref{Lemma4}. The above inequality implies, because of Gronwall's Lemma, 
$$
\mathcal{E}_N(t)\le e^{\left[\left(2Md_0 +2^{M-1}3c_0\right) t\right]}\mathcal{E}_N(0)\, .
$$
If the initial state of the system is such that condition in Eq.~\eqref{theorem1-2} is satisfied, we  have
$$
\lim_{N\to\infty}\mathcal{E}_N(t)\le e^{\left(2Md_0 +2^{M-1}3c_0\right) t}\lim_{N\to\infty}\mathcal{E}_N(0)=0\, .
$$

Now, because of the bounds 
\begin{equation}
\begin{split}
\left|\omega_t\left(m_{ b,k}^N-m_{ b,k}(t)\right)\right|  \le \sqrt{\mathcal{E}_N(t)}\, ,\qquad \qquad
\left|\omega_t\left(\alpha_{\mu,N}-\alpha_\mu(t)\right)\right|  \le \sqrt{\mathcal{E}_N(t)}\, ,
\end{split}
\label{theorem-partial3}
\end{equation}
the statement in Eq.~\eqref{theorem1-1} implies that 
\begin{equation}
\begin{split}
\lim_{N\to\infty}\omega_t\left(m_{ b,k}^N\right)=m_{ b,k}(t)\, ,\qquad \qquad 
\lim_{N\to\infty}\omega_t\left(\alpha_{\mu,N}\right)=\alpha_\mu(t)\, .
\end{split}
\label{theorem-partial4}
\end{equation}
Physically, these relations mean that the mean-field dynamical equations \eqref{lim-eqs} correctly capture the time-evolution of macroscopic spin and boson operators, in the limit $N\to\infty$. 
\end{proof}

\section*{LEMMATA}

\begin{lemma} Given the sequences of operators defined in Eq.~\eqref{seq-spin}, the scalar time-dependent functions of \eqref{lim-eqs}, and the quantity $\mathcal{E}_N(t)$, defined in Eq.~\eqref{error} and given explicitely in Eq.~\eqref{def-E}, the following bounds hold:
\begin{align}
\left|\omega_t\left(\left[m_{ b,k}^N-m_{ b,k}(t)\right]\left[m_{ a,h}^N-m_{ a,h}(t)\right]\right)\right|&\le \mathcal{E}_N(t)\, ,
\label{Lemma1-1}\\
\left|\omega_t\left(\left[\alpha_{\mu,N}-\alpha_\mu(t)\right]^\dagger m_{ b,k}^N\left[m_{ a,h}^N-m_{ a,h}(t)\right]\right)\right|&\le \mathcal{E}_N(t)\, ,
\label{Lemma1-2}\\
\left|\omega_t\left(\left[\alpha_{\mu,N}-\alpha_\mu(t)\right] m_{ b,k}^N\left[m_{ a,h}^N-m_{ a,h}(t)\right]\right)\right|&\le \mathcal{E}_N(t)\, ,
\label{Lemma1-3}\\
\left|\omega_t\left(\left[m_{ a,k}^N-m_{ a,k}(t)\right]\left[\alpha_{\mu,N}-\alpha_\mu(t)\right]\right)\right|&\le \mathcal{E}_N(t)\, ,
\label{Lemma1-4}
\end{align}
\label{Lemma1}
\end{lemma}
\begin{proof}
Before considering all different cases, we derive a bound which is valid for generic operators. We will then show, case by case via an appropriate choice of the operators, each of the above relations. \\

\noindent We start considering the expectation value $\omega_t\left(A^\dagger C B\right)$. The state $\omega_t$ is a positive linear and normalized functional. Thus, we can use Cauchy-Schwarz inequality to obtain
\begin{equation}
\left|\omega_t\left(A^\dagger C B\right)\right|\le \sqrt{\omega_t\left(A^\dagger A\right)}\sqrt{\omega_t\left(B^\dagger C^\dagger C B\right)}\, .
\end{equation}
In addition, the inequality $C^\dagger C\le \|C\|^2$ implies that
$$
\omega_t\left(B^\dagger C^\dagger C B\right)\le \|C\|^2\omega_t\left(B^\dagger B\right)\, .
$$
Altogether, we have
\begin{equation}
\left|\omega_t\left(A^\dagger C B\right)\right|\le \|C\|\sqrt{\omega_t\left(A^\dagger A\right)}\sqrt{\omega_t\left(B^\dagger  B\right)}\, .
\label{Lemma-partial1}
\end{equation}
To proceed we make the following observation. Consider two positive numbers $x,y$: the square of their difference is a positive number 
$$
(x-y)^2=x^2+y^2-2xy\ge0\, .
$$
By turning the above relation around, we get
$$
xy\le \frac{1}{2}\left(x^2+y^2\right)\le x^2 +y^2\, ,
$$
where the second inequality is obtained by removing the factor $1/2$. Identifying $x=\sqrt{\omega_t\left(A^\dagger A\right)}$ and $y=\sqrt{\omega_t\left(B^\dagger B\right)}$, this means that  
$$
\sqrt{\omega_t\left(A^\dagger A\right)}\sqrt{\omega_t\left(B^\dagger B\right)}\le \omega_t\left(A^\dagger A\right)+\omega_t\left(B^\dagger B\right)\, .
$$
Using the above finding in Eq.~\eqref{Lemma-partial1}, we have the relation
\begin{equation}
\left|\omega_t\left(A^\dagger C B\right)\right|\le \|C\|\left(\omega_t\left(A^\dagger A\right)+\omega_t\left(B^\dagger B\right)\right)\, ,
\label{Lemma-partial2}
\end{equation}
which we use to show all relations in the statement of the Lemma. \\

To prove the Eq.\eqref{Lemma1-1}, we consider the bound in Eq.~\eqref{Lemma-partial2} with $C={\bf 1}$, $A=m_{ b,k}^N-m_{ b,k}(t)$ and $B=m_{ a,h}^N-m_{ a,h}(t)$, and then add on the right hand side of the resulting Eq.~\eqref{Lemma-partial2} all the missing terms to reconstruct $\mathcal{E}_N(t)$. This can be done since each term forming $\mathcal{E}_N(t)$ is positive. Notice that if ${ b=a}$ and ${ k=h}$, then equation \eqref{Lemma1-1} is trivially satisfied. \\

To prove Eq.~\eqref{Lemma1-2}, we consider Eq.~\eqref{Lemma-partial2} with $C=m_{ b,k}^N$, noticing that $\|m_{ b,k}^N\|=1$. Then, we take $A^\dagger =\left(\alpha_{\mu,N} -\alpha_\mu(t)\right)^\dagger$ and $B=m_{ a,h}^N-m_{ a,h}(t)$, and add on the right hand side of the resulting Eq.~\eqref{Lemma-partial2} the remaining terms.\\

For Eq.~\eqref{Lemma1-3}, we have the same as above but with $A^\dagger = \alpha_{\mu,N}-\alpha_\mu(t)$. \\

Finally, for Eq.~\eqref{Lemma1-4} we take $A=m_{ a,k}^N-m_{ a,k}(t)$, $B=\alpha_{\mu,N}-\alpha_\mu(t)$ and $C={\bf 1}$ and proceed as above. 
\end{proof}

\begin{lemma} Given the sequences of operators $m_{ a,k}^N$ and $\alpha_{\mu,N}$ defined in Eq.~\eqref{seq-spin}, and the sequence of dynamical generators $\mathcal{L}_N$ introduced in the main text, we have that 
\begin{align*}
\mathcal{L}_N[m_{ a,k}^N]&=-2\Omega\sum_{ b=x,y,z}\epsilon_{ xab}m_{ b,k}^N-2g\sum_{\mu=1}^M\sum_{ b=x,y,z}f_{\mu,{ k}}^M\left(\alpha_{\mu,N}+\alpha_{\mu,N}^\dagger\right)\epsilon_{ zab}m_{ b,k}^N\\
\mathcal{L}_N[\alpha_{\mu,N}]&=-\left(i\Omega_\mu+\frac{\kappa_\mu}{2}\right)\alpha_{\mu,N}-ig\sum_{ k=1}^{2^M-1}f_{\mu,{ k}}^M m_{ z,k}^N\\
\end{align*}
\label{Lemma2}
\end{lemma}

\begin{proof} To obtain the relations above, it is sufficient to compute the action of the generator on the considered operators, using their definition and exploiting the  commutation relations of the bosonic operators and the algebraic rules
$$
[S_{a,k},S_{b,h}]=i2\delta_{ k,h}\sum_{ c=x,y,z}\epsilon_{ abc}\, S_{ c,k}\, .
$$
The tensor $\epsilon_{ abc}$ is the fully antisymmetric tensor, while the Kronecker delta appears because operators of the subset $ k$ commute with operators of the subset $ h$. 
\end{proof}

\begin{lemma} Given the sequences of operators $m_{ a,k}^N$ defined in Eq.~\eqref{seq-spin}, the sequence of dynamical generators $\mathcal{L}_N$ introduced in the main text and the scalar time-dependent functions in \eqref{lim-eqs}, we have that 
\begin{align*}
\left|\frac{d}{dt}\omega_t\left(\left[m_{ a,k}^N-m_{ a,k}(t)\right]^2\right)\right|\le c_0\, \mathcal{E}_N(t)
\end{align*}
where
$$
c_0=12\left|\Omega\right|+24M\left|g\right|+12\left|g\right|M\bar{\beta}\, ,
$$
and 
$$
\bar{\beta}=2\max_{\forall \mu}\left(\left|\alpha_\mu(0)\right|+\sqrt{3}\frac{2^M\left|g\right|}{\kappa_\mu}\right)\, .
$$
\label{Lemma3}
\end{lemma}

\begin{proof}
First of all, we recall that the scalar functions $m_{ a,k},\alpha_\mu$ are solution to the mean-field equations \eqref{lim-eqs}, which are obtained via the factorization assumption of expectation values of the Heisenberg equations emerging from the results presented in Lemma \ref{Lemma2}. With this in mind, we start defining 
\begin{equation}
D_t:=\frac{d}{dt}\omega_t\left(\left[m_{ a,k}^N-m_{ a,k}(t)\right]^2\right)\, ,
\label{Lemma3-partial}
\end{equation}
and consider explicitly the time-derivative:
$$
D_t=\omega_t\left(\mathcal{L}_N\left[\left[m_{ a,k}^N-m_{ a,k}(t)\right]^2\right]\right)-2\dot{m}_{ a,k}(t)\, \omega_t\left(m_{a,k}^N-m_{ a,k}(t)\right)\, .
$$
The first term on the right-hand side of the above equation is obtained by taking the derivative on the state functional $\omega_t$ and using Eq.~\eqref{t-der-state}; the second term is, instead, emerging from the time-derivative applied to the term in the square brackets of Eq.~\eqref{Lemma3-partial}. 

Focussing on the action of the Lindblad map on the operator, we have that   
$$
\mathcal{L}_N\left[\left[m_{ a,k}^N-m_{a,k}(t)\right]^2\right]=i\left[H_N,\left[m_{ a,k}^N-m_{ a,k}(t)\right]^2\right]=\mathcal{L}_N\left[m_{ a,k}^N\right]\left[m_{ a,k}^N-m_{ a,k}(t)\right]+\left[m_{ a,k}^N-m_{ a,k}(t)\right]\mathcal{L}_N\left[m_{ a,k}^N\right]\, ,
$$
which follows from the fact that the generator annihilates the term proportional  to the identity ($\mathcal{L}_N[m_{ a,k}(t)]=0$) and that $\mathcal{L}_N$ acts directly on spins only via a Hamiltonian term. 

We thus write
\begin{equation}
\begin{split}
D_t&=\omega_t\left(\mathcal{L}_N\left[m_{ a,k}^N\right]\left[m_{ a,k}^N-m_{ a,k}(t)\right]\right)-\omega_t\left(\dot{m}_{ a,k}(t)\left[m_{ a,k}^N-m_{ a,k}(t)\right]\right)+\\
&+\omega_t\left(\left[m_{ a,k}^N-m_{ a,k}(t)\right]\mathcal{L}_N\left[m_{ a,k}^N\right]\right)-\omega_t\left(\left[m_{ a,k}^N-m_{ a,k}(t)\right]\dot{m}_{ a,k}(t)\right)\, .
\end{split}
\end{equation}
Notice that the terms $m_{ a,k}(t),\dot{m}_{ a,k}(t)$ can be safely pulled inside or outside of the state expectation, since they are scalar quantities. Collecting the first two terms of the above equation, as well as the second two terms, we obtain
$$
D_t=\omega_t\left(\left[\mathcal{L}_N\left[m_{ a,k}^N\right]-\dot{m}_{ a,k}(t)\right]\left[m_{ a,k}^N-m_{ a,k}(t)\right]\right)+\omega_t\left(\left[m_{ a,k}^N-m_{ a,k}(t)\right]\left[\mathcal{L}_N\left[m_{ a,k}^N\right]-\dot{m}_{a,k}(t)\right]\right)\, .
$$
Since the second term of the above equation is the complex conjugate of the first one, we just focus on the latter. We define it as 
$$
D_t^I=\omega_t\left(\left[\mathcal{L}_N\left[m_{ a,k}^N\right]-\dot{m}_{ a,k}(t)\right]\left[m_{ a,k}^N-m_{ a,k}(t)\right]\right)\, .
$$
Exploiting Lemma \ref{Lemma2} and the differential equations \eqref{lim-eqs}, we have that (we leave summation indeces implicit)
\begin{equation}
\mathcal{L}_N\left[m_{ a,k}^N\right]-\dot{m}_{ a,k}(t)=-2\Omega \sum_{ b}\epsilon_{ xab}\left[m_{ b,k}^N-m_{ b,k}(t)\right]-2g\sum_{\mu,{ b}}\epsilon_{ zab}f_{\mu,k}^M\left[\left(\alpha_{\mu,N}^\dagger +\alpha_{\mu,N} \right)m_{ b,k}^N-\left(\alpha_\mu^\dagger(t)+\alpha_\mu(t)\right)m_{ b,k}(t)\right]\, . 
\label{Lemma3-partial2}
\end{equation}

We need to reshape the last contributions to the above equation in a way that we can rewrite them in terms of the difference between macroscopic operators and their corresponding scalar values. To this end, we consider that 
$$
\left(\alpha_{\mu,N}^\dagger +\alpha_{\mu,N}\right)m_{ b,k}^N-\left(\alpha_\mu^\dagger(t) +\alpha_\mu(t)\right)m_{ b,k}(t)\!=\!\left(\!\alpha_{\mu,N}^\dagger\!+\!\alpha_{\mu,N}-\alpha_{\mu}^\dagger(t) -\alpha_\mu(t)\!\right)\!m_{ b,k}^N+\left(\alpha_\mu^\dagger(t) +\alpha_\mu(t)\right)\!\left(m_{ b,k}^N-m_{ b,k}(t)\right),
$$
where we have simply added and substracted the term $\left(\alpha_\mu^\dagger(t)+\alpha_\mu(t)\right)m_{ b,k}^N$.  From the above relation, we obtain 
\begin{equation}
\begin{split}
\left(\alpha_{\mu,N}^\dagger +\alpha_{\mu,N}\right)m_{ b,k}^N-\left(\alpha_\mu^\dagger(t) +\alpha_\mu(t)\right)m_{ b,k}(t)&=\left[\alpha_{\mu,N}^\dagger -\alpha_\mu^\dagger(t)\right]m_{ b,k}^N+\left[\alpha_{\mu,N}-\alpha_\mu(t)\right]m_{ b,k}^N+\\
&+\left(\alpha_\mu^\dagger(t) +\alpha_\mu(t)\right)\left[m_{ b,k}^N-m_{ b,k}(t)\right]\, .
\end{split}
\label{Lemma3-partial3}
\end{equation}
By substituting this in the square brackets of Eq.~\eqref{Lemma3-partial2}, we can write 
\begin{equation*}
\begin{split}
\mathcal{L}_N\left[m_{ a,k}^N\right]-\dot{m}_{ a,k}(t)=&-2\Omega \sum_{ b}\epsilon_{ xab}\left[m_{ b,k}^N-m_{ b,k}(t)\right]-2g\sum_{\mu,{ b}}\epsilon_{ zab}f_{\mu,{ k}}^M\left[\alpha_{\mu,N}^\dagger -\alpha_{\mu}^\dagger (t)\right]m_{ b,k}^N+\\
&-2g\sum_{\mu,{ b}}\epsilon_{ zab}f_{\mu,{ k}}^M \left[\alpha_{\mu,N}-\alpha_\mu(t)\right]m_{ b,k}^N-2g\sum_{\mu,{ b}}\epsilon_{ zab}f_{\mu,{ k}}^M\left(\alpha^\dagger_\mu(t)+\alpha_\mu(t)\right)\left[m_{ b,k}^N-m_{ b,k}(t)\right]\, , 
\end{split}
\end{equation*}
and use it to find the expression for $D_t^I$,
\begin{equation*}
\begin{split}
D_t^I&=-2\Omega \sum_{ b}\epsilon_{ xab} \omega_t\left(\left[m_{ b,k}^N-m_{ b,k}(t)\right]\left[m_{ a,k}^N-m_{ a,k}(t)\right]\right)\!-\!2g\sum_{\mu, { b}}\epsilon_{ zab}f_{\mu,{ k}}^M\omega_t\left(\left[\alpha_{\mu,N} -\alpha_\mu(t)\right]^\dagger m_{ b,k}^N\left[m_{ a,k}^N-m_{ a,k}(t)\right]\right)+\\
&-2g\!\sum_{\mu,{ b}}\!\epsilon_{ zab}f_{\mu,{ k}}^M\omega_t\left(\left[\alpha_{\mu,N}-\alpha_\mu(t)\right]m_{ b,k}^N\left[m_{ a,k}^N-m_{ a,k}(t)\right]\right)\!+\\
&-\!2g\!\sum_{\mu,{ b}}\!\epsilon_{ zab}f_{\mu,{ k}}^M\left(\alpha^\dagger_\mu(t)+\alpha_\mu(t)\right)\omega_t\left(\left[m_{ b,k}^N-m_{ b,k}(t)\right]\left[m_{ a,k}^N-m_{ a,k}(t)\right]\right).
\end{split}
\end{equation*}
The task is now to find proper bounds for each of these terms. We consider that
\begin{equation*}
\begin{split}
\left|D_t^I\right|&\le 2\left|\Omega\right| \sum_{ b}\left| \omega_t\left(\left[m_{ b,k}^N-m_{ b,k}(t)\right]\left[m_{ a,k}^N-m_{ a,k}(t)\right]\right)\right|+2\left|g\right|\sum_{\mu, { b}}\left|\omega_t\left(\left[\alpha_{\mu,N} -\alpha_\mu(t)\right]^\dagger m_{ b,k}^N\left[m_{ a,k}^N-m_{ a,k}(t)\right]\right)\right|+\\
&+2\left|g\right|\sum_{\mu,{ b}}\left|\omega_t\left(\left[\alpha_{\mu,N}-\alpha_\mu(t)\right]m_{ b,k}^N\left[m_{ a,k}^N-m_{ a,k}(t)\right]\right)\right|+\\
&+2\left|g\right|\sum_{\mu,{ b}}\left|\alpha^\dagger_\mu(t)+\alpha_\mu(t)\right|\left|\omega_t\left(\left[m_{ b,k}^N-m_{ b,k}(t)\right]\left[m_{ a,k}^N-m_{ a,k}(t)\right]\right)\right|\, ;
\end{split}
\end{equation*}
for the expectation value in the first and the last term we can use Eq.~\eqref{Lemma1-1} in Lemma \ref{Lemma1}, while for the expectation value in the second term we use Eq.~\eqref{Lemma1-2} and for the third Eq.~\eqref{Lemma1-3}. Considering also the extensions of the summations, we obtain the following bound 
\begin{equation}
\left|D_t^I\right|\le \left(6\left|\Omega\right|+6M\left|g\right|+6M\left|g\right|+6M\left|g\right|\beta_t\right)\mathcal{E}_N(t)\, .
\label{Lemma3-partial4}
\end{equation}
In the above relation we have introduced the quantity $\beta_t$ defined as  
$$
\beta_t:=\sup_{\forall\mu}\left|\alpha_\mu(t)+\alpha_\mu^\dagger(t)\right|\, ,
$$
for $t\ge0$,  needed to bound the scalar term in the fourth term of Eq.~\eqref{Lemma3-partial4}. To achieve a meaningful bound we need to show that $\beta_t$ is finite. To this end, we take advantage of the formal solution of the mean-field equations \eqref{lim-eqs} to get 
\begin{equation*}
\alpha_\mu(t)=e^{-\left(i\Omega_\mu+\kappa_\mu/2\right)t}\alpha_\mu(0)-ig\int_0^t ds \, e^{-\left(i\Omega_\mu+\kappa_\mu/2\right)(t-s)}\sum_{ k=1}^{2^{M-1}}f_{\mu, { k}}^M \, m_{ z,k}(s)\, .
\end{equation*}
We then take the modulus of the above relation 
\begin{equation*}
\left|\alpha_\mu(t)\right|\le e^{-t\, \kappa_\mu/2}\left|\alpha_\mu(0)\right|+\left|g\right|\int_0^t ds \, e^{-(t-s)\kappa_\mu/2}\sum_{ k=1}^{2^{M-1}}\left|m_{ z,k}(s)\right|\, .
\end{equation*}
We notice that 
$$
\left|m_{ z,k}(t)\right|=\sqrt{m_{ z,k}^2(t)}\le \sqrt{m_{ T,k}^2}\, ,
$$
where we have $m_{ T,k}^2=m_{ x,k}^2(t)+m_{ y,k}^2(t)+m_{ z,k}^2(t)$. For Eqs.~\eqref{lim-eqs}, $m_{ T,k}^2$ is a constant of motion. Given that the initial values for the system of differential equations \eqref{lim-eqs} are to be taken as 
$$
m_{ a,k}(0)=\lim_{N\to\infty}\omega\left(m_{ a,k}^N\right)\le 1\, ,
$$
the initial value for $m_{ T,k}$ is such that $|m_{ T,k}|\le\sqrt{3}$. This is a loose bound, since for physical reasons one would expect $|m_{ T,k}|\le1$. However, without assumptions on the initial state $\omega$, the bound $|m_{ T,k}|\le\sqrt{3}$ is more readily found. Thus, using that $|m_{ z,k}(t)|\le \sqrt{3}$, for all times $t\ge0$, we have 
\begin{equation*}
\left|\alpha_\mu(t)\right|\le \left|\alpha_\mu(0)\right|+\sqrt{3}\frac{2^M\left|g\right|}{\kappa_\mu}\, ,
\end{equation*}
where we also bounded the integral. This shows that 
$$
\beta_t\le 2\max_{\forall \mu}\left(\left|\alpha_\mu(0)\right|+\sqrt{3}\frac{2^M\left|g\right|}{\kappa_\mu}\right)=:\bar{\beta}\, ,
$$
which is finite if the initial $\alpha_\mu$'s have finite modulus. 

So far we have found that $\left|D_t^I\right|\le c_0/2\mathcal{E}_N(t)$
with $c_0$ as defined in the statement of the Lemma. We can conclude the proof by noticing that 
$$
\left|D_t\right|\le \left|D_t^I\right|+\left|\left(D_t^I\right)^*\right|\le c_0\, \mathcal{E}_N(t)\, .
$$
\end{proof}

\begin{lemma} Given the sequences of operators defined in Eq.~\eqref{seq-spin}, the sequence of dynamical generators $\mathcal{L}_N$ introduced in the main text, and the scalar time-dependent functions in \eqref{lim-eqs}, we have that 
\begin{align}
\left|\frac{d}{dt}\omega_t\left(\left[\alpha_{\mu,N}-\alpha_\mu(t)\right]^\dagger\left[\alpha_{\mu,N}-\alpha_\mu(t)\right]\right)\right|&\le d_0\mathcal{E}_N(t)\, ,
\label{Lemma4-1}\\
\left|\frac{d}{dt}\omega_t\left(\left[\alpha_{\mu,N}-\alpha_\mu(t)\right]\left[\alpha_{\mu,N}-\alpha_\mu(t)\right]^\dagger\right)\right|&\le d_0\mathcal{E}_N(t)\, ,
\label{Lemma4-2}
\end{align}
where 
$$
d_0=2\left[\Gamma+\left(2^M-1\right)\left|g\right|\right]\, ,
$$
and 
$$
\Gamma=\max_{\forall \mu}\left(\left|\Omega_\mu\right|+\frac{\kappa_\mu}{2}\right)
$$
\label{Lemma4}
\end{lemma}

\begin{proof} First, we focus on the proof of Eq.~\eqref{Lemma4-1}. We define and compute explicitly the time-derivative 
\begin{equation}
\begin{split}
\tilde{D}_t:&=\frac{d}{dt}\omega_t\left(\left[\alpha_{\mu,N}-\alpha_\mu(t)\right]^\dagger\left[\alpha_{\mu,N}-\alpha_\mu(t)\right]\right)\\
&=\omega_t\left(\mathcal{L}_N\left[\left[\alpha_{\mu,N}-\alpha_\mu(t)\right]^\dagger\left[\alpha_{\mu,N}-\alpha_\mu(t)\right]\right]\right)-\dot{\alpha}_\mu^\dagger (t)\omega_t\left(\alpha_{\mu,N}-\alpha_{\mu}(t)\right)-\dot{\alpha}_\mu(t)\omega_t\left(\alpha_{\mu,N}^\dagger-\alpha_\mu^\dagger(t)\right)\, , 
\end{split}
\label{tildeD}
\end{equation}
where we have used the relation in Eq.~\eqref{t-der-state} and the fact that $\alpha_\mu$ is a time-dependent scalar quantity. To proceed we need to consider the action of the Lindblad generator on the operators. We note that 
$$
\mathcal{L}_N\left[X^\dagger X\right]=\mathcal{L}_N\left[X^\dagger \right]X+X^\dagger \mathcal{L}_N\left[X\right]+\sum_{\nu}\kappa_\nu\left[a^\dagger_\nu,X^\dagger\right]\left[X,a_\nu\right]\, ,
$$
and, since in our case $X=\alpha_{\mu,N}-\alpha_\mu(t)$, the third term on the right hand side of the above relation is not contributing. As such, we can write
$$
\mathcal{L}_N\left[\left[\alpha_{\mu,N}-\alpha_\mu(t)\right]^\dagger\left[\alpha_{\mu,N}-\alpha_\mu(t)\right]\right]=\mathcal{L}_N\left[\alpha_{\mu,N}^\dagger\right]\left[\alpha_{\mu,N}-\alpha_\mu(t)\right]+\left[\alpha_{\mu,N}-\alpha_\mu(t)\right]^\dagger \mathcal{L}_N\left[\alpha_{\mu,N}\right]\, .
$$
Introducing this in the time-derivative of Eq.~\eqref{tildeD} we have 
\begin{equation}
\begin{split}
\tilde{D}_t&=\omega_t\left(\mathcal{L}_N\left[\alpha_{\mu,N}^\dagger\right]\left[\alpha_{\mu,N}-\alpha_\mu(t)\right]\right)-\dot{\alpha}_\mu^\dagger (t)\omega_t\left(\alpha_{\mu,N}-\alpha_\mu(t)\right)\\
&+\omega_t\left(\left[\alpha_{\mu,N}-\alpha_\mu(t)\right]^\dagger \mathcal{L}_N\left[\alpha_{\mu,N}\right]\right)-\dot{\alpha}_\mu(t)\omega_t\left(\left[\alpha_{\mu,N}-\alpha_\mu(t)\right]^\dagger\right)\\
&=\omega_t\left(\left[\mathcal{L}_N\left[\alpha_{\mu,N}^\dagger\right]-\dot{\alpha}_\mu^\dagger(t)\right]\left[\alpha_{\mu,N}-\alpha_\mu(t)\right]\right)+\omega_t\left(\left[\alpha_{\mu,N}-\alpha_\mu(t)\right]^\dagger\left[\mathcal{L}_N\left[\alpha_{\mu,N}\right]-\dot{\alpha}_\mu(t)\right]\right)\, ,
\end{split}
\label{Lemma4-partial2}
\end{equation}
where for the second equality we grouped the first and the second terms and the third and the fourth ones appearing after the first equality in Eq.~\eqref{Lemma4-partial2}. This can be done given that $\alpha_\mu(t)$ is a scalar and can be moved inside and outside of the expectation over states without problems. 

Considering that the second term in the second line of Eq.~\eqref{Lemma4-partial2} is the complex conjugate of the first, we define 
\begin{equation}
\tilde{D}_t^I:=\omega_t\left(\left[\mathcal{L}_N\left[\alpha_{\mu,N}^\dagger\right]-\dot{\alpha}_\mu^\dagger(t)\right]\left[\alpha_{\mu,N}-\alpha_\mu(t)\right]\right)\, ,
\label{Lemma4-partial3}
\end{equation}
so that 
$$
\tilde{D}_t=\tilde{D}_t^I+\left(\tilde{D}_t^I\right)^*\, ,
$$
and we can focus on $\tilde{D}_t^I$. To this end, we look at the term in the first round bracket in Eq.~\eqref{Lemma4-partial3}: we have 
$$
\mathcal{L}_N\left[\alpha_{\mu,N}^\dagger \right]-\dot{\alpha}_\mu^\dagger(t)=-\left(-i\Omega_\mu+\frac{\kappa_\mu}{2}\right)\left[\alpha_{\mu,N} -\alpha_\mu(t)\right]^\dagger +ig\sum_{ k}f_{\mu,{ k}}^M\left[m_{ z,k}^N-m_{ z,k}(t)\right]\, .
$$
Inserting the above equation back into Eq.~\eqref{Lemma4-partial3} we obtain 
$$
\tilde{D}_t^I=\left(i\Omega_\mu-\frac{\kappa_\mu}{2}\right)\omega_t\left(\left[\alpha_{\mu,N}-\alpha_\mu(t)\right]^\dagger \left[\alpha_{\mu,N}-\alpha_\mu(t)\right] \right)+ig\sum_{ k}f_{\mu,{ k}}^M \, \omega_t\left(\left[m_{ z,k}^N-m_{ z,k}(t)\right]\left[\alpha_{\mu,N}-\alpha_\mu(t)\right]\right)\, .
$$
We can now proceed to bound the term $\tilde{D}_t^I$ and we get
$$
\left|\tilde{D}_t^I\right|\le \left(\left|\Omega_\mu\right|+\frac{\kappa_\mu}{2}\right)\omega_t\left(\left[\alpha_{\mu,N}-\alpha_\mu(t)\right]^\dagger \left[\alpha_{\mu,N}-\alpha_\mu(t)\right]\right)+\left|g\right|\sum_{ k}\left|\omega_t\left(\left[m_{ z,k}^N-m_{ z,k}(t)\right]\left[\alpha_{\mu,N}-\alpha_\mu(t)\right]\right)\right|\, ;
$$
the first expectation on the right-hand side of the above equation is smaller than $\mathcal{E}_N(t)$. For the second term, we can use Eq.~\eqref{Lemma1-4}  in Lemma~\ref{Lemma1}. All together, considering also that the sum is over $2^{M-1}$ terms, this leads to 
\begin{equation}
\left|\tilde{D}_t^I\right|\le \left(\Gamma +|g|2^{M-1}\right)\mathcal{E}_N(t)
\label{Lemma4-partial4}
\end{equation}
where we have introduced the term $\Gamma=\max_{\forall \mu}\left(\left|\Omega_\mu\right|+\kappa_\mu/2\right)$. We can now use Eq.~\eqref{Lemma4-partial4} to achieve the bound 
$$
\left|\tilde{D}_t\right|\le d_0\, \mathcal{E}_N(t)\, ,
$$
with $d_0=2\left(\Gamma+|g|2^{M-1}\right)$, which proves the first relation [Eq.~\eqref{Lemma4-1}] of the Lemma. 

For the second relation we proceed as follows. Given the commulation relations of the $\alpha_{\mu,N}$, we have 
$$
\omega_t\left(\left[\alpha_{\mu,N}-\alpha_\mu(t)\right]\left[\alpha_{\mu,N}-\alpha_\mu(t)\right]^\dagger\right)=\omega_t\left(\left[\alpha_{\mu,N}-\alpha_\mu(t)\right]^\dagger\left[\alpha_{\mu,N}-\alpha_\mu(t)\right]\right)+\frac{1}{N}\, .
$$
We can thus relate the time-derivative in Eq.~\eqref{Lemma4-2} to the one in Eq.~\eqref{Lemma4-1} as follows 
$$
\frac{d}{dt}\omega_t\left(\left[\alpha_{\mu,N}-\alpha_\mu(t)\right]\left[\alpha_{\mu,N} -\alpha_\mu(t)\right]^\dagger\right)=\frac{d}{dt}\omega_t\left(\left[\alpha_{\mu,N} -\alpha_\mu(t)\right]^\dagger\left[\alpha_{\mu,N} -\alpha_\mu(t)\right]\right)\, ,
$$
and we can use this relation together with the previous result to find the bound in Eq.~\eqref{Lemma4-2}.
\end{proof}

\section*{Stationary solution with one finite overlap}
In this section, we provide details on the computation showing that a stationary solution to Eqs.~\eqref{lim-eqs} featuring a finite overlap with one of the patterns exists. To simplify the notation, we consider 
\begin{equation*}
m_{ x,k}=x_{ k}\, ,\qquad m_{ y,k}=y_{ k}\, ,\qquad m_{ z,k}=z_{ k}\, .
\end{equation*}
As reported in the main text, we want to show that a stationary solution to the mean-field equations, with $z_{ k}=f_{ \nu,k}^M|z|$, 
exists. In particular, this form implies a finite overlap with the pattern $\nu$. Indeed, we have 
$$
\xi_\mu=\frac{1}{2^{M-1}}\sum_{ k=1}^{2^{M-1}}f_{ \mu,k}^Mz_{ k}=\frac{|z|}{2^{M-1}}\sum_{ k=1}^{2^{M-1}}f_{ \mu,k}^Mf_{ \nu,k}^M=|z|\delta_{ \mu,\nu}\, ,
$$
since the quantity 
\begin{equation}
\sum_{ k=1}^{2^{M-1}}f_{ \mu,k}^Mf_{ \nu,k}^M=2^{M-1}\delta_{\mu,\nu}\, .
\label{f-f}
\end{equation}

\noindent For completeness, we explicitely write the equations of motion
\begin{equation}
\begin{split}
\dot{x}_{ k}&=-2g\sum_{\mu=1}^Mf_{ \mu,k}^M\left(\alpha_\mu^\dagger+\alpha_\mu\right)y_{ k}\, ,\\
\dot{y}_{ k}&=-2\Omega z_{ k}+2g\sum_{\mu=1}^Mf^M_{ \mu,k}\left(\alpha_\mu^\dagger+\alpha_\mu\right)x_{ k}\, ,\\
\dot{z}_{ k}&=2\Omega y_{ k}\, ,\\
\dot{\alpha}_{\mu}&=-\left(i\Omega_\mu+\frac{\kappa_{\mu}}{2}\right)\alpha_\mu-ig\sum_{ k=1}^{2^{M-1}}f_{ \mu,k}^M \, z_{ k}\, .
\end{split}
\label{eq-explicit}
\end{equation}
To look for stationary solutions, we need to set each of the above equations to zero. We assume $y_{ k}=0$, for all ${ k}$: this takes care of the first and the third equations. Then, we take $z_{ k}=f_{ \nu,k}^M|z|$ and consider that, for all ${ k}$, the initial total angular momentum is equal $m_{ T,k}=s$. Since this is a conserved quantity, $m_{ T,k}$ can be used to provide a relation between $x_{ k}$ and $z_{ k}$, at stationarity. In particular, we have 
\begin{equation}
x_{ k}=\pm\sqrt{s^2-|z|}\, .
\label{x-z-rel}
\end{equation}
We now look at the fourth equation. Setting this to zero, we obtain
$$
\alpha_\mu=2^{M-1}g|z|\delta_{\mu,\nu}\frac{-\Omega_\mu-i\frac{\kappa_\mu}{2}}{\Omega_\mu^2+\left(\frac{\kappa_\mu}{2}\right)^2}\, ,
$$
and thus 
$$
\left(\alpha_\mu+\alpha_\mu^\dagger\right)=-2^{M}g|z|\delta_{\mu,\nu}\frac{\Omega_\mu}{\Omega_\mu^2+\left(\frac{\kappa_\mu}{2}\right)^2}\, .
$$
We can now exploit this result for the second equation in \eqref{eq-explicit}. For $|z|\neq0$, we find 
$$
-2\Omega-g^22^{M+1}\frac{\Omega_\nu}{\Omega_\nu^2+\left(\frac{\kappa_\nu}{2}\right)^2}x_{ k}=0\, .
$$
Without loss of generality we assume all coefficients to be positive. In this case, $x_{ k}$ must be negative and using equation \eqref{x-z-rel} we have
$$
|z|^2=s^2-\frac{1}{4g_0^4}\left(\frac{\Omega}{\Omega_\nu}\right)^2\left[\Omega_\nu^2+\left(\frac{\kappa_\nu}{2}\right)^2\right]^2\, ,
$$
where we have further considered that $g=g_0/\sqrt{2^{M-1}}$. This relation can be satisfied only when the right hand side is positive. For concreteness, we take $s=1$. This means that the proposed solution is possible only if
$$
\frac{1}{4g_0^4}\left(\frac{\Omega}{\Omega_\nu}\right)^2\left[\Omega_\nu^2+\left(\frac{\kappa_\nu}{2}\right)^2\right]^2\le 1\, .
$$
The critical $g_0$, i.e.~the $g_0$ making the above relation an equality, is given by 
$$
g_0=\sqrt{\frac{1}{2}\left(\frac{\Omega}{\Omega_\nu}\right)\left[\Omega_\nu^2+\left(\frac{\kappa_\nu}{2}\right)^2\right]}\, .
$$

\end{document}